 \documentclass[letterpaper, 10 pt, conference]{ieeeconf}  

\IEEEoverridecommandlockouts                              
\usepackage{comment}
\usepackage{setspace}
\usepackage{amsmath}
\usepackage{amssymb}
\usepackage{amsthm}
\usepackage{etoolbox}
\usepackage{color,colortbl}
\usepackage{mathtools}
\usepackage{enumerate}
\usepackage{algorithm, algpseudocode}
\usepackage{booktabs}
\usepackage[normalem]{ulem}
\usepackage{pgf}
\usepackage{tikz,pgfplots}
\usetikzlibrary{trees,decorations,arrows,arrows.meta,automata,shadows,positioning,plotmarks,backgrounds,shapes,shapes.misc}
\usetikzlibrary{calc,matrix,fit,petri,decorations.pathmorphing,patterns}
\usetikzlibrary{decorations.pathreplacing,decorations.markings,shapes.geometric,calc}
\usepackage{paralist}
\usepackage{stmaryrd}
\usepackage{xspace}
\usepackage{graphicx}
\usepackage{float}
\usepackage[hang, tight]{subfigure}
\usepackage[utf8]{inputenc} 
 \usepackage[english]{babel} 
  \usepackage{csquotes} 
  \usepackage{epstopdf}

 \usepackage{array}

\usepackage{xifthen}
\usepackage{url}
\usepackage{relsize}

\usepackage{macros}

\newcommand\set[1]{\{ #1 \}}
\newcommand\tuple[1]{{\langle #1 \rangle}}

\newcommand{\REFlem}[1]{\text{Lem.~\ref{#1}}}
\newcommand{\REFthm}[1]{\text{Thm.~\ref{#1}}}

\newcommand{\REFalg}[1]{Alg.~\ref{#1}}

\newcommand{\REFsec}[1]{Sec.~\ref{#1}}

\newcommand{\REFfig}[1]{Fig.~\ref{#1}}

\newcommand{\Cpre}{\mathrm{CPre}}
\newcommand{\FCpre}[1]{\Cpre_{\Sa{#1}}}

\newcommand{\UA}[1]{\Gamma_{#1}}

\newcommand{\ON}[1]{\operatorname{#1}}

\def\clap#1{\hbox to 0pt{\hss#1\hss}}

\makeatletter 
\newif\ifFIRST
\newif\ifSECOND
\let\LISTOP\relax
\newcommand{\List}[4][\;]{#3#1%
	\FIRSTtrue
	\@for\i:=#2\do{%
	\ifFIRST\LISTOP{\i}\FIRSTfalse\else,\LISTOP{\i}\fi%
	}%
	#1#4%
	\let\LISTOP\relax
}
\makeatother
\newcounter{DINGLIST}
\stepcounter{DINGLIST}
\newcommand{\markD}[3][\;\;]{\text{\ding{\the\numexpr171+\theDINGLIST}\stepcounter{DINGLIST}}#1#3}

\makeatletter

\newcommand{\propNeg}{\@ifstar\propNegStar\propNegNoStar}
\newcommand{\propNegStar}[1]{\ensuremath{\left(\propNegNoStar{#1}\right)}}
\newcommand{\propNegNoStar}[2][\cdot]{\ensuremath{\neg\ifthenelse{\isempty{#2}}{#1}{#2}}}

\newcommand{\propConj}{\@ifstar\propConjStar\propConjNoStar}
\newcommand{\propConjStar}[2]{\ensuremath{\left(\propConjNoStar{#1}{#2}\right)}}
\newcommand{\propConjNoStar}[3][\cdot]{\ensuremath{\ifthenelse{\isempty{#2}}{#1}{#2}\wedge\ifthenelse{\isempty{#3}}{#1}{#3}}}

\newcommand{\propDisj}{\@ifstar\propDisjStar\propDisjNoStar}
\newcommand{\propDisjStar}[2]{\ensuremath{\left(\propDisjNoStar{#1}{#2}\right)}}
\newcommand{\propDisjNoStar}[3][\cdot]{\ensuremath{\ifthenelse{\isempty{#2}}{#1}{#2}\vee\ifthenelse{\isempty{#3}}{#1}{#3}}}

\newcommand{\propImp}{\@ifstar\propImpStar\propImpNoStar}
\newcommand{\propImpStar}[2]{\ensuremath{\left(\propImpNoStar{#1}{#2}\right)}}
\newcommand{\propImpNoStar}[3][\cdot]{\ensuremath{\ifthenelse{\isempty{#2}}{#1}{#2}\Rightarrow\ifthenelse{\isempty{#3}}{#1}{#3}}}

\newcommand{\propAequ}{\@ifstar\propAequStar\propAequNoStar}
\newcommand{\propAequStar}[2]{\ensuremath{\left(\propAequNoStar{#1}{#2}\right)}}
\newcommand{\propAequNoStar}[3][\cdot]{\ensuremath{\ifthenelse{\isempty{#2}}{#1}{#2}\Leftrightarrow\ifthenelse{\isempty{#3}}{#1}{#3}}}

\newcommand{\AllQ}{\@ifstar\AllQStar\AllQNoStar}
\newcommand{\AllQStar}[3][\;]{\ensuremath{\left(\forall #2#1.#1#3\right)}}
\newcommand{\AllQNoStar}[3][\;]{\ensuremath{\forall #2#1.#1#3}}
\newcommand{\AllQu}{\@ifstar\AllQuStar\AllQuNoStar}
\newcommand{\AllQuStar}[3][\;]{\ensuremath{\left(\forall^{\infty} #2#1.#1#3\right)}}
\newcommand{\AllQuNoStar}[3][\;]{\ensuremath{\forall^{\infty} #2#1.#1#3}}

\newcommand{\ExQ}{\@ifstar\ExQStar\ExQNoStar}
\newcommand{\ExQStar}[3][\;]{\ensuremath{\left(\exists #2#1.#1#3\right)}}
\newcommand{\ExQNoStar}[3][\;]{\ensuremath{\exists #2#1.#1#3}}

\newcommand{\NExQ}{\@ifstar\NExQStar\NExQNoStar}
\newcommand{\NExQStar}[3][\;]{\ensuremath{\left(\nexists #2#1.#1#3\right)}}
\newcommand{\NExQNoStar}[3][\;]{\ensuremath{\nexists #2#1.#1#3}}

\newcommand{\UniqueQ}{\@ifstar\UniqueQStar\UniqueQNoStar}
\newcommand{\UniqueQStar}[3][\;]{\ensuremath{\left(\exists! #2#1.#1#3\right)}}
\newcommand{\UniqueQNoStar}[3][\;]{\ensuremath{\exists! #2#1.#1#3}}

  \newlength{\SFS@HEIGHT}
  \newlength{\SFS@WIDTH}
  \newcommand{\SplitX}[2]{
	    \settoheight{\SFS@HEIGHT}{$#2$}
	    \settowidth{\SFS@WIDTH}{$#2$}
	    \mbox{\begin{tikzpicture}[baseline=(current bounding box.center)]
	    \node[] (E) at (0,0) {$#1$};
	    \node[inner sep=0pt] (F) at ($(E.south west)+(1ex,-1ex)+(3ex+.5\SFS@WIDTH,-\SFS@HEIGHT)$) {$#2$};
	    \node[] (E) at (0,0) {\phantom{$#1$}};
	    \draw[fill] ($(E.east)+(1ex,0ex)$) circle (.2ex);
	    \draw[-] ($(E.east)+(1ex,0ex)$) -- ($(E.south east)+(1ex,-0.5ex)$) -- ($(E.south west)+(1ex,-0.5ex)$) -- ($(E.south west)+(1ex,-1ex)-(0,\SFS@HEIGHT)$) -- ($(E.south west)+(2.5ex,-1ex)-(0,\SFS@HEIGHT)$);
	    \draw[fill] ($(E.south west)+(2.5ex,-1ex)-(0,\SFS@HEIGHT)$) circle (.2ex);
	    \end{tikzpicture}}}
  \newcommand{\SplitS}[2]{
	    \settoheight{\SFS@HEIGHT}{$#2$}
	    \settowidth{\SFS@WIDTH}{$#2$}
	    \mbox{\begin{tikzpicture}[baseline=(current bounding box.center)]
	    \node[] (E) at (0,0) {$#1$};
	    \node[inner sep=0pt] (F) at ($(E.south west)+(1ex,0.5ex)+(3ex+.5\SFS@WIDTH,-\SFS@HEIGHT)$) {$#2$};
	    \end{tikzpicture}}}

\newcommand{\Set}[2][]{\List[#1]{#2}{\left\{}{\right\}}}
\newcommand{\VSet}[2][]{\let\LISTOP\val\List[#1]{#2}{\{}{\}}}

\newcommand{\VTuple}[2][]{\let\LISTOP\val\List[#1]{#2}{(}{)}}

\newcommand{\UNION}{\@ifstar\UNIONStar\UNIONNoStar}
\newcommand{\UNIONStar}[2]{\ensuremath{\left(\UNIONNoStar{#1}{#2}\right)}}
\newcommand{\UNIONNoStar}[2]{\ensuremath{\ifthenelse{\isempty{#1}}{\cdot}{#1}\cup\ifthenelse{\isempty{#2}}{\cdot}{#2}}}

\newcommand{\UNIOND}{\@ifstar\UNIONDStar\UNIONDNoStar}
\newcommand{\UNIONDStar}[2]{\ensuremath{\left(\UNIONDNoStar{#1}{#2}\right)}}
\newcommand{\UNIONDNoStar}[2]{\ensuremath{\ifthenelse{\isempty{#1}}{\cdot}{#1}\uplus\ifthenelse{\isempty{#2}}{\cdot}{#2}}}

\newcommand{\SETMINUS}{\@ifstar\SETMINUSStar\SETMINUSNoStar}
\newcommand{\SETMINUSStar}[2]{\ensuremath{\left(\SETMINUSNoStar{#1}{#2}\right)}}
\newcommand{\SETMINUSNoStar}[2]{\ensuremath{\ifthenelse{\isempty{#1}}{\cdot}{#1}\setminus\ifthenelse{\isempty{#2}}{\cdot}{#2}}}

\newcommand{\INTERSECT}{\@ifstar\INTERSECTStar\INTERSECTNoStar}
\newcommand{\INTERSECTStar}[2]{\ensuremath{\left(\INTERSECTNoStar{#1}{#2}\right)}}
\newcommand{\INTERSECTNoStar}[2]{\ensuremath{\ifthenelse{\isempty{#1}}{\cdot}{#1}\cap\ifthenelse{\isempty{#2}}{\cdot}{#2}}}

\newcommand{\CARTPROD}{\@ifstar\CARTPRODStar\CARTPRODNoStar}
\newcommand{\CARTPRODStar}[2]{\ensuremath{\left(\CARTPRODNoStar{#1}{#2}\right)}}
\newcommand{\CARTPRODNoStar}[2]{\ensuremath{\ifthenelse{\isempty{#1}}{\cdot}{#1}\times\ifthenelse{\isempty{#2}}{\cdot}{#2}}}

\newcommand{\FINCOUNT}{\@ifstar\FinCountStar\FinCountNoStar}
\newcommand{\FinCountStar}[1]{\ensuremath{\#(\ifthenelse{\isempty{#1}}{\cdot}{#1})}}
\newcommand{\FinCountNoStar}[1]{\ensuremath{\#\left(\ifthenelse{\isempty{#1}}{\cdot}{#1}\right)}}

\makeatother

\newcommand{\real}[1]{\ifstrempty{#1}{\mathbb{R}}{\mathbb{R}^{#1}}}
\newcommand{\Z}{\mathbb{Z}}

\newcommand{\fun}{\ensuremath{\ON{\rightarrow}}}
\newcommand{\SetComp}[3][]{\{#1#2#1\mid#1#3#1\}}

\newcommand{\twoup}[1]{\ensuremath{2^{#1}}} 

\newcommand{\dom}[1]{\ensuremath{\mathrm{dom}(#1)}}

\newcommand{\Beh}[1]{\ensuremath{\mathcal{B}(#1)}}
\newcommand{\img}[1]{\ensuremath{\mathrm{img}(#1)}}

 \newcommand{\Behaclset}{\ensuremath{\mathcal{B}(\Saset^{cl})}}
 \newcommand{\Behtclset}{\ensuremath{\mathcal{B}(\Stset^{cl})}}

\newcommand{\hyint}[1]{\ensuremath{\llbracket #1\rrbracket}}
\newcommand{\tn}[1]{{#1}}
\newcommand{\nindex}[1]{{#1}}
\newcommand{\n}[1]{{\eta_{#1}}}
\newcommand{\ta}[1]{{\tau_{#1}}}
\newcommand{\tindex}[1]{{#1}}
\newcommand{\frr}[1]{\preccurlyeq_{#1}}
\newcommand{\layer}{\ensuremath{{l\in[1;L]}}\xspace}

\newcommand{\St}[1]{\ensuremath{\overrightarrow{S}_{\tindex{#1}}}}
\newcommand{\Ft}[1]{\ensuremath{\overrightarrow{F}_{\tindex{#1}}}}

\newcommand{\Sa}[1]{\ensuremath{\widehat{S}_{\tn{#1}}}}

\newcommand{\Xa}[1]{\ensuremath{\widehat{X}_{\nindex{#1}}}}

\newcommand{\Ua}{\ensuremath{\widehat{U}}}
\newcommand{\Fa}[1]{\ensuremath{\widehat{F}_{\tn{#1}}}}

\newcommand{\Qa}[1]{\ensuremath{\widehat{Q}_{\nindex{#1}}}}
\newcommand{\Qai}[1]{\ensuremath{\widehat{Q}^{-1}_{\nindex{#1}}}}

\newcommand{\Ra}[1]{\ensuremath{\widehat{R}_{#1}}}

\newcommand{\xa}{\ensuremath{\widehat{x}}}

\newcommand{\xia}{\ensuremath{\widehat{\xi}}}
\newcommand{\xit}{\ensuremath{\overrightarrow{\xi}}}
\newcommand{\ua}{\ensuremath{\widehat{u}}}

\newcommand{\Ci}[1]{\ensuremath{C^{#1}}}
\newcommand{\Uci}[1]{\ensuremath{B^{#1}}}
\newcommand{\Gci}[1]{\ensuremath{G^{#1}}}

\newcommand{\Uc}{\ensuremath{B}}

\newcommand{\Xaall}{\ensuremath{\widehat{\textbf{X}}}}

\newcommand{\WIN}{\ensuremath{\mathcal{C}}}

\newcommand{\Aux}{\ensuremath{\Upsilon}}
\newcommand{\Cset}{\ensuremath{\mathbf{C}}}

\newcommand{\Stset}{\ensuremath{\overrightarrow{\mathbf{S}}}}
\newcommand{\Saset}{\ensuremath{\widehat{\mathbf{S}}}}

\newcommand{\Qset}{\ensuremath{\mathbf{Q}}}

\newcommand{\Ftall}{\ensuremath{\overrightarrow{\mathbf{F}}}}
\newcommand{\Faall}{\ensuremath{\widehat{\mathbf{F}}}}

\newcommand{\Tset}{\ensuremath{T}}

\newcommand{\Tseta}[1]{\ensuremath{\widehat{\Tset}_{#1}}}

\newcommand{\Ya}[1]{\ensuremath{\widehat{Y}_{\tn{#1}}}}

\theoremstyle{definition}
\newtheorem{theorem}{Theorem}

\newtheorem{lemma}{Lemma}

\title{\LARGE \bf
Lazy Abstraction-Based Control for Safety Specifications
}

\author{Kyle Hsu, Rupak Majumdar, Kaushik Mallik, Anne-Kathrin Schmuck
\thanks{K.~Hsu is with University of Toronto, Canada and R.~Majumdar, K.~Mallik and A.-K.~Schmuck are with MPI-SWS, Kaiserslautern, Germany
        {\tt\small \{kylehsu,rupak,kmallik,akschmuck\}@mpi-sws.org}}%
}

\begin{document}

\maketitle
\thispagestyle{empty}
\pagestyle{empty}

\begin{abstract}
We present a \emph{lazy} version of multi-layered abstraction-based controller synthesis (ABCS) for continuous-time nonlinear dynamical
systems against safety specifications.
State-of-the-art multi-layered ABCS uses pre-computed finite-state abstractions of different coarseness.
Our new algorithm improves this technique by computing transitions on-the-fly,
and only when a particular region of the state space needs to be explored by the controller synthesis algorithm for
a specific coarseness.
Additionally, our algorithm improves upon existing techniques by using coarser 
cells on a larger subset of the state space, which leads to significant computational savings.
%

\end{abstract}


\setlength{\textfloatsep}{5pt}

\section{Introduction}

Abstraction-based controller synthesis (ABCS) is a general three-step procedure for the automatic synthesis of
controllers for non-linear dynamical systems \cite{TabuadaBook, Girard12}. 
First, a time-sampled version of the continuous dynamics of the open-loop system is abstracted by a symbolic finite state model. 
Second, algorithms from reactive synthesis are used to synthesize a discrete controller on the abstract system. 
Third, the abstract controller is refined to a controller for the concrete system. 

The abstract system can be constructed by fixing a parameter $\tau$ for the sample time and
a parameter $\eta$ for the state and input spaces. 
The abstract state space is then represented as a set of hypercubes, each of diameter $\eta$,
and the abstract transition relation is constructed by adding a transition between two hypercubes iff there exists some state
in the first hypercube which can reach some state of the second by following the original dynamics for time $\tau$.
This construction establishes a \emph{feedback refinement relation} (FRR) \cite{ReissigWeberRungger_2017_FRR} between the concrete system and the abstract system which is commonly used to prove soundness of ABCS. 

The success of ABCS depends on the choice of $\eta$ and $\tau$. 
A large $\eta$ (and $\tau$)\footnote{
	$\tau$ is increased along with $\eta$ to reduce non-determinism due to self loops.}
results in an imprecise abstract transition relation with a small state space, while a small $\eta$ (and $\tau$) results in a precise abstraction with a large state space.
Thus, for a large $\eta$ one may not be able to find a controller while a small $\eta$ can make the synthesis problem computationally intractable.
Thus, recent approaches to ABCS use a \emph{multi-layered} technique, where one constructs several
``layers'' of abstractions using hypercube partitions defined by progressively larger $\eta$ and $\tau$ \cite{CameraGirardGoessler_safety_2011,CameraGirardGoessler_reach_2011,girard2006towards,Girard2016_InterSampling,HsuMajumdarMallikSchmuck_HSCC18}.
Here, the controller synthesis procedure tries to find a controller for the coarsest abstraction whenever 
feasible, but adaptively considers finer abstractions when necessary. 
The common bottleneck of these approaches is that the full abstract transition system for every granularity needs to be pre-computed.
In this paper, we propose a \emph{lazy} multi-layered algorithm that
reduces this computational overhead by computing transitions on-the-fly, and
only when a particular region of the state space needs to be explored by the controller synthesis algorithm 
for a specific choice of $\eta$ and $\tau$. 

We start with a backward symbolic algorithm for safety control, \`a la reactive synthesis.
We use the multi-layered $\omega$-regular synthesis approach of Hsu et al.~\cite{HsuMajumdarMallikSchmuck_HSCC18},
but improve upon that algorithm by interleaving fixpoint computations in different abstraction layers.
Theoretically, safe states in finer layers can be used when iterating in coarser layers.
Empirically, this allows the algorithm to use coarser cells on a larger subset of the state space.

\subsection{Motivating Example}\label{sec:spiral}
We show the advantage of our algorithm over the one in \cite{HsuMajumdarMallikSchmuck_HSCC18} using an example.
Consider a simple dynamical system in polar coordinates:
\begin{align}\label{equ:spiral}
	\dot{r} = -0.1r + u		&&	\dot{\theta} = 1
\end{align}
where $r$ and $\theta$ represent the radius and the angle respectively, and $u$ represents the control input. The resulting dynamics generate a circular motion of its trajectories in a two-dimensional Cartesian state space, where the input controls the radius of this motion. The control problem is to avoid the static obstacles in the state space, depicted in black in \REFfig{fig:spiral}. An efficient multi-layered safety controller synthesis algorithm would use coarse grid cells almost everywhere in the state space and would use finer grid cells only close to the obstacles, where the control action needs to be precise.

While the idea is conceptually simple, the implementation is challenging due to the following observations. To ensure safety, one wants to find the largest invariant set within the safe set. To obtain the described behavior, this invariant set needs to consist of cells with different coarseness. To compute this using established abstraction-refinement techniques as in, e.g., \cite{AlfaroRoy_2010}, one needs a common game graph representation connecting states of different coarseness. However, due to the absence of an FRR between different layers of abstraction (see \cite{HsuMajumdarMallikSchmuck_HSCC18} for an in-depth discussion of this issue) and the use of different sampling times for different layers, we do not have such a representation. We can therefore only run iterations of the safety fixed-point for a particular layer, but not for combinations of them.

In \cite{HsuMajumdarMallikSchmuck_HSCC18} this problem is circumvented by computing the safety fixed-point for all layers until termination, starting with the coarsest. This implies that coarse grid cells are only used by the resulting controller if they form an invariant set among themselves. For our toy example, this corresponds to the small green region depicted in \REFfig{fig:spiral} (top).

We improve upon this result by proposing a new algorithm in this paper which keeps iterating over all layers until termination of the fixed-point. 
This has the effect that clusters of finer cells which can be controlled to be safe by a suitable controller in the corresponding layer
are considered safe in the coarser layers as well, which influences further iterations of the safety fixed-point in those layers. 
This results in the desired behavior for this example shown in \REFfig{fig:spiral} (bottom), 
where almost the whole controller domain is covered by the coarsest layer cells (depicted in green) and finer layers are only used around the obstacles and at the boundary of the safe set. 
\begin{figure}
\centering
 \includegraphics[width=\columnwidth]{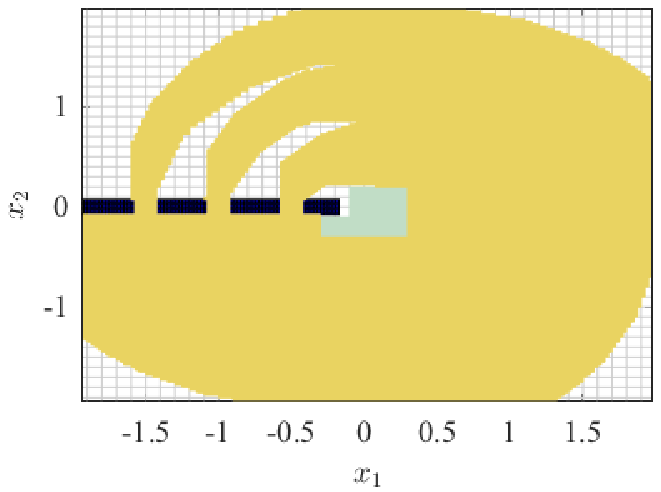}
%
 \includegraphics[width=1\columnwidth]{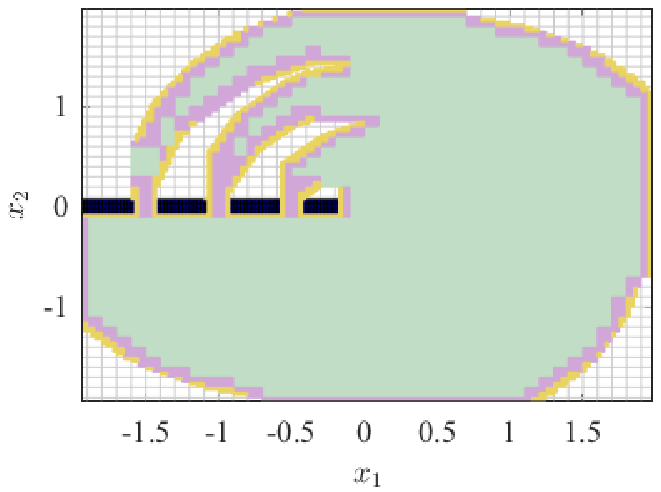}
 \caption{Resulting controller domains for layer $l=1$ (yellow), $l=2$ (magenta) and $l=3$ (green) computed for the dynamics in \eqref{equ:spiral} using the algorithm in \cite{HsuMajumdarMallikSchmuck_HSCC18}, Sec.~4.3 (top) and the new algorithm presented in \REFsec{sec:ML_Safety_Frontier} of this paper (bottom). Large $l$ is coarser. Obstacles are depicted in black.}\label{fig:spiral}
\end{figure}
As expected, this leads to computational savings; our new algorithm runs 2x faster on this example than does the algorithm presented in \cite{HsuMajumdarMallikSchmuck_HSCC18}. Section~\ref{sec:experiments} also shows computational savings of our new algorithm on the DC-DC boost converter benchmark example.

\subsection{Related Work}\label{sec:related}
The multi-layered ABCS algorithm for safety specifications by Girard et al. \cite{Girard2016_InterSampling} considers a strict subclass of 
the dynamics considered in this paper and uses a modified version of approximate bisimulation relations instead of the more general FRR considered in this paper. This results in a deterministic abstract model, which allows for a forward search based technique to synthesize safety controllers. 
While forward search is usually faster for safety, it is not known how to symbolically handle external disturbances and non-determinism 
in the abstraction in a forward algorithm.

Meyer et al. propose an abstraction refinement algorithm for a fragment of LTL specifications \cite{meyer2017abstraction}, where a nominal single integrator model is used to find a viable plan which is used to refine abstractions locally along the planned path. As the nominal model is not based on the system dynamics, this heuristic may or may not give quick convergence. Our approach does not suffer from this problem.

Nilsson et al. propose an abstraction refinement technique to synthesize switching protocols for switched systems and reach-avoid-stay specifications \cite{nilsson2014incremental}. The ``stay'' part of their algorithm solves a safety game while adaptively refining the abstraction as needed. This algorithm suffers from the same problem as \cite{HsuMajumdarMallikSchmuck_HSCC18} (see \REFsec{sec:spiral}).

%

\section{Preliminaries}\label{sec:prelim}


\noindent\textbf{Notation.}
Given $a,b\in\real{}\cup\set{\pm\infty}$ with $a < b$, we denote by $[a,b]$ a closed interval. 
Given $a,b\in(\real{}\cup\set{\pm\infty})^n$, we denote by $a_{i}$ and $b_{i}$ their $i$-th element. 
A \emph{cell} $\hyint{a,b}$ with $a<b$ (component-wise) is the closed set $\set{x\in\real{}^n\mid a_i\leq x_i\leq b_i}$.
We define the relations $<,\leq,\geq,>$ on $a,b$ component-wise. 
For a set $A$, we write $A^*$ and $A^\infty$ for the set of finite, and the set of finite or infinite sequences over $A$, respectively.
For $w\in A^*$, we write $|w|$ for the length of $w$; the length of an infinite sequence is $\infty$. 
For $0 \leq k < |w|$ we write $w(k)$ for the $k$-th symbol of $w$.

\smallskip
\noindent\textbf{Continuous-Time Control Systems.}\
A \emph{control system} $\Sigma = (X, U, W, f)$
consists of a state space $X= \real{n}$, a non-empty input space $U\subseteq\real{m}$, 
a compact set $W\subset \real{n}$,
and a function $f:X\times U \rightarrow \real{n}$ locally Lipschitz in the first argument s.t.
\begin{equation}\label{equ:def_f}
 \dot{\xi}\in f(\xi(t),u(t))+W
\end{equation}
holds.  
Given an initial state $\xi(0)\in X$, a positive parameter $\tau>0$ and a constant input trajectory $\mu_u:[0,\tau]\rightarrow U$
which maps every $t\in [0,\tau]$ to the same $u\in U$, 
a solution of the inclusion in \eqref{equ:def_f} 
on $[0,\tau]$ is an absolutely continuous function $\xi:[0,\tau]\rightarrow X$  
that fulfills \eqref{equ:def_f} for almost every $t\in[0,\tau]$. 
We collect all such solutions in the set $\ON{Sol}_f(\xi(0),\tau,u)$. 

\smallskip
\noindent\textbf{Time-Sampled System.}\
Given a time sampling parameter $\tau>0$, we define the \emph{time-sampled system} $\St{}(\Sigma,\tau)=(X,U,\Ft{})$ associated with $\Sigma$, 
where $\Ft{}:X\times U\fun 2^X$ is the transition function, defined s.t.\ for 
all $x\in X$ and for all $u \in U$ it holds that $x'\in \Ft{}(x,u)$ iff there exists a solution $\xi\in\ON{Sol}_f(x,\tau,u)$ 
s.t.\ $\xi(\tau)=x'$.
A \emph{trajectory} $\xit$ of $\St{}(\Sigma,\tau)$ is a finite or infinite sequence $x_0\xrightarrow{u_0}x_1\xrightarrow{u_1} \ldots$
such that for each $i\geq 0$, $x_{i+1}\in\Ft{}(x_i, u_i)$; the collection of all such trajectories defines the behavior $\Beh{\St{}(\Sigma,\tau)}\subseteq X^\infty$.

\smallskip
\noindent\textbf{Abstract Systems.}\ 
A \emph{cover} $\hat{X}$ of  $X$ is a set of non-empty cells $\hyint{a,b}$ with $a,b\in (\real{}\cup\Set{\pm\infty})^n$,
s.t.\ every $x\in X$ belongs to some cell $\xa\in\hat{X}$. 
We fix a grid parameter $\eta \in\real{}_{>0}^n$ and a global safety requirement
$Y = \hyint{\alpha, \beta}$, s.t.\ $\beta - \alpha$ is an integer multiple of $\eta$.
A point $c\in Y$ is \emph{grid-aligned} if there is $k\in\Z^n$ s.t. for each $i\in [1;n]$,
$c_i = \alpha_i + k_i\eta_i - \frac{\eta_i}{2}$.
A cell $\hyint{a,b}$ is \emph{grid-aligned} if there is a grid-aligned point $c$ s.t.\ $a = c - \frac{\eta}{2}$ and
$b = c + \frac{\eta}{2}$;
such cells define sets of diameter $\eta$ whose center-points are grid aligned. 
Clearly, the set of grid-aligned cells is a \emph{finite cover} for $Y$.

We define an \emph{abstract system} $\Sa{}(\Sigma,\tau,\eta)=(\Xa{},\Ua{},\Fa{})$ s.t.\ the following holds:
\begin{inparaenum}[(i)]
 \item $\Xa{}$ is a finite cover of $X$ and there exists a non-empty subset 
 $\Ya{}\subseteq\Xa{}$ which is a cover of $Y$ with grid aligned cells,
 \item $\Ua\subseteq U$ is finite,
 \item $\Fa{}:\Xa{}\times \Ua{}\rightarrow 2^{\Xa{}}$ is the transition function s.t. for all $\xa\in(\Xa{}\setminus\Ya{})$ and $u\in\Ua$ it holds that $\Fa{}(\xa,u)=\emptyset$, and
 \item for all $\xa\in\Ya{}$, $\xa'\in\Xa{}$, and $u\in\Ua{}$ it holds that\footnote{We use the technique explained in \cite{ReissigWeberRungger_2017_FRR} and implemented in 
\texttt{SCOTS} \cite{Scots} to over-approximate the set $\set{\cup_{x\in\xa}\ON{Sol}_f(x,\tau,\ua)}$ in \eqref{eq:next state abs sys 0}.} 
 \begin{align}\label{eq:next state abs sys 0}
   \propAequ{\xa'\in\Fa{}(\xa,u)}{\set{\cup_{x\in\xa}\ON{Sol}_f(x,\tau,\ua)} \cap \xa' \neq \emptyset}.
 \end{align}
\end{inparaenum}

We consider multiple abstract systems obtained in this way.
For parameters $\n{1} > 0$ and $\ta{1} >0$, 
and for $l\in\mathbb{Z}_{>1}$, we define $\n{l} = 2\n{l-1}$ and $\ta{l} = 2\ta{l-1}$.
With this, we obtain a sequence of $L$ time-sampled systems $\Stset:=\set{\St{l}(\Sigma,\ta{l})}_\layer$
and $L$ abstract systems $\Saset:=\set{\Sa{l}(\Sigma,\ta{l},\n{l})}_\layer$ with $\St{l}=(X,U,\Ft{l})$ and $\Sa{l}=(\Xa{l},\Ua{},\Fa{l})$\footnote{If $\Sigma$, $\tau$, and $\eta$ are
clear from the context, we omit them in $\St{l}$ and $\Sa{l}$.}.

\smallskip
\noindent\textbf{Feedback Refinement Relations.} 
Let $\Qa{}\subseteq X\times \Xa{}$ be a relation s.t.\ $(x,\xa)\in\Qa{}$ iff $x\in\xa$. Then $\Qa{}$ is a \emph{feedback refinement relation (FRR)} from $\St{}$ to $\Sa{}$ 
written $\St{}\frr{\Qa{}}\Sa{}$ (see \cite{ReissigWeberRungger_2017_FRR}, Thm. III.5).
That is, $\Qa{}$ is a strict relation,
i.e., for each $x$, there is some $\xa$ such that $(x,\xa)\in \Qa{}$, and 
for all $(x,\xa)\in \Qa{}$, we have
\begin{inparaenum}[(i)]
 \item $U_{\Sa{}}(\xa)\subseteq U_{\St{}}(x)$, and 
 \item $u\in U_{\Sa{}}(\xa) \Rightarrow \Qa{}(\Ft{}(x,u))\subseteq \Fa{}(\xa,u)$,
\end{inparaenum}
 where $U_{\St{}}(x):=\SetComp{u\in U}{\Ft{}(x,u)\neq \emptyset}$ and $U_{\Sa{}}(x):=\SetComp{u\in U}{\Fa{}(x,u)\neq \emptyset}$.

For $\Stset$ and $\Saset$, we have a sequence $\set{\Qa{l}}_\layer$ of FRRs between the corresponding systems.
The set of FRRs $\set{\Qa{l}}_\layer$ induces transformers\footnote{We extend $\Qa{}$ and $\Ra{}$ to sets of states in the obvious way.} 
$\Ra{ll'}\subseteq\Xa{l}\times\Xa{l'}$ for 
$1\leq l$, $l'\leq L$ between abstract states of different layers such that
\begin{align}\label{equ:Ra}
 \propAequ{\xa\in\Ra{ll'}(\xa')}{\xa\in \Qa{l}(\Qai{l'}(\xa')).}
\end{align}
Note that $\Ra{ll'}$ is generally \emph{not} an FRR between the layers.
%

\smallskip
\noindent\textbf{Multi-Layered Controllers and Closed Loops.}
Given a multi-layered abstract system $\Saset$, a multi-layered controller is defined as  $\Cset=\set{\Ci{l}}_{\layer}$ where for all $l\in [1;L]$ $\Ci{l}=(\Uci{l}, \Ua,\Gci{l})$, $\Uci{l}$ is the controller domain, and $\Gci{l}:\Uci{l}\fun 2^{\Ua}$ is the feedback control map. $\Cset$ is composable with $\Saset$ if $\Ci{l}$ is composable with $\Sa{ l}$ for all $\layer$, i.e., $\Uci{l}\subseteq \Xa{ l}$. We denote by $\dom{\Cset}=\bigcup_\layer\Uci{l}$ the domain of $\Cset$.
Given a multi-layered controller $\Cset$, we define the \emph{quantizer induced by $\Cset$} as the 
map $\Qset:X\fun 2^{\Xaall}\setminus\set{\emptyset}$ with $\Xaall=\bigcup_\layer \Xa{l}$ s.t.\ for all $x\in X$ it holds that $\xa\in\Qset(x)$ iff 
 there exists $\layer$ s.t.\
  $\xa\in\Qa{l}(x)\cap\Uci{l}$
 and there exists no $l'>l$ and 
  $\xa'\in\Qa{l'}(x)\cap\Uci{l'}$.
%
We define $\img{\Qset}=\SetComp{\xa\in\Xaall}{\ExQ{x\in X}{\xa\in\Qset(x)}}$. 
Intuitively, $\Qset$ maps states $x\in X$ to the \emph{coarsest} abstract state $\xa$ that is both related to $x$ and in the domain of $\Cset$.

The \emph{closed loop system} formed by interconnecting
$\Saset$ and $\Cset$ in \emph{feedback} 
is defined by the system 
$\Saset^{cl}=(\Xaall, \Faall^{cl})$ with $\Faall^{cl}:\img{\Qset}\fun2^{\img{\Qset}}$ s.t.\ 
$\xa'\in\Faall^{cl}(\xa)$ iff there exists $\layer$, $\ua\in\Gci{l}(\xa)$ and $\xa''\in\Fa{ l}(\xa,\ua)$ s.t.\ $\xa'\in\Qset(\Qai{ l}(\xa''))$.
As $\St{l}\frr{\Qa{l}} \Sa{l}$ for all $l\in[1;L]$, $\Cset$
can be refined into a controller composable with $\Stset$ using $\Qset$ (see \cite{HsuMajumdarMallikSchmuck_HSCC18}, Sec.~3.4). This results in the closed loop system
$\Stset^{cl}=(X, \Ftall^{cl})$ with $\Ftall^{cl}:X\fun 2^X$ s.t.\
$x'\in\Ftall^{cl}(x)$ iff there exists $\xa\in \Qset(x)$, $\layer$ and $\ua\in\Gci{l}(\xa)$ s.t.\ $x'\in\Ft{ l}(x,\ua)$.
The behavior of $\Stset^{cl}$ and $\Saset^{cl}$ are defined by 
 \begin{align*}
 \Behtclset&:=\textstyle\SetComp{\xi\in X^\infty}{\AllQ{1\leq k < |\xi|}{\xi(k)\in \Ftall^{cl}(\xi(k-1))}}\\
 \Behaclset&:=\textstyle\SetComp{\xia\in \Xaall^\infty}{\AllQ{1\leq k < {|\xia|}}{\xia(k)\in \Faall^{cl}(\xia(k-1))}}.
\end{align*}
 Note that $\Behaclset$ contains trajectories composed from abstract states of different coarseness and $\Behtclset$ contains trajectories with non-uniform sampling time.

\section{Problem Statement}\label{sec:SafetyProblems}
\label{sec:safety control}

Consider a set of safe states $\Tset{}\subseteq Y\subseteq X$, where $Y$ is the global safety requirement used to construct the \emph{finite} abstractions $\set{\Sa{l}}_\layer$ (see \REFsec{sec:prelim}). Then we consider the safety control problem $\tuple{\Sigma,\Tset{}}$ which
asks for a controller to be constructed such that every trajectory $\xi$ of the closed loop system stays within $\Tset{}$ at \emph{sampling instances}.

A multi-layered controller $\Cset$ therefore solves $\tuple{\Sigma,\Tset{}}$ if for all $\xit\in\Behtclset$ holds that $\xit(k)\in\Tset{}$ for all $k\in\dom{\xit}$. Note that in this case the considered sampling instances might be non-uniformly spaced. By adopting a classical result of ABCS using FRR (see \cite[Sec.VI.A]{ReissigWeberRungger_2017_FRR}) to the multi-layered case (see \cite{HsuMajumdarMallikSchmuck_HSCC18}, Sec.~3.4) we know that $\Cset$ solves $\tuple{\Sigma,\Tset{}}$ in this sense, if for abstract trajectories of the closed loop formed by $\Cset$ and $\Saset{}$ holds that
\begin{equation}\label{equ:Csetsound}
 \AllQ{\xia\in\Behaclset}{\AllQ{k\in\dom{\xia}}{\Qset^{-1}(\xia(k))\subseteq\Tset{}}}.
\end{equation}
When considering an under-approximation of $\Tset{}$ by $\Tseta{l}\subseteq\Ya{l}$ for every \layer s.t.\ $\Tseta{l}=\SetComp{\xa\in\Xa{l}}{\xa\subseteq \Tset{}}$, the right side of \eqref{equ:Csetsound} states that $\xia(k)\in\img{\Qset}\cap\Xa{l}$ implies $\xia(k)\in\Tseta{l}$. That is, if $\xia(k)$ is a layer $l$ cell which is currently the largest cell in the domain of $\Cset$, then it must be contained in the under-approximation $\Tseta{l}$ of the safe set. 
We collect all multi-layered controllers which solve $\tuple{\Sigma,\Tset{}}$ in the set $\WIN(\Sigma,\Tset{})$.

It is common practice in ABCS to ensure safety for sampling times only. This implicitly assumes that sampling times and grid sizes are chosen such that no \enquote{holes} occur between consecutive cells visited in a trajectory. This can be formalized by additional assumptions on the growth rate of $f$ in \eqref{equ:def_f} which is beyond the scope of this paper.

\section{Abstraction-Based Safety Control}\label{sec:ML_Safety}

This section presents non-lazy abstraction-based safety control before its lazy version is introduced in \REFsec{sec:ML_Safety_Frontier}.

\smallskip
\noindent\textbf{Single-Layered Control}
We consider a 
safety control problem $\tuple{\Sigma,\Tset{}}$ and recall how it is commonly solved by ABCS for $L=l=1$.
In this case one iteratively computes the sets
\begin{equation}\label{equ:safe-fp}
	W^0 = \Tseta{l} \text{ and } W^{i+1} = \FCpre{l}(W^i) \cap \Tseta{l}
\end{equation}
until an iteration $N\in \mathbb{N}$ with $W^N = W^{N+1}$ is reached, where 
$\FCpre{l}:\twoup{\Xa{l}}\fun\twoup{\Xa{l}}$ is the \emph{controllable predecessor} operator, defined for a set $\Aux\subseteq\Xa{l}$ by
\begin{equation}
 \FCpre{l}(\Aux) := \set{\xa\in\Xa{l} \mid \exists \ua\in\Ua{}\;.\;\Fa{l}(\xa,\ua) \subseteq \Aux}. \label{eq:define cpre}\\
\end{equation} 
Then $C = (B,\Ua{},G)$ with $B =W^N$, and
\begin{equation}\label{equ:safe-controller}
\propImp{\ua\in\Gci{l}(\xa)}{\Fa{l}(\xa,\ua)\subseteq B}
\end{equation}
for all $\xa\in B$, is known to be a safety controller for $\tuple{\Sigma,\Tset{}}$. 

Thus, the above synthesis algorithm is sound.
However, completeness is not guaranteed; there may exist a state $x\in X$ s.t. $\Qa{l}(x)\not\in \Uc$, and there may exist a controller $C' = (\Uc',\Ua{},G')$, solving the safety control problem $\tuple{\Sigma,\Tset{}}$ s.t.\ $\Qa{l'}(x)\in \Uc'$ with $l'<l$.

\smallskip
\noindent\textbf{Multi-layered Control}
Given a sequence of $L$ abstract systems $\Saset:=\set{\Sa{l}}_\layer$ we now present a non-lazy multi-layered safety algorithm formalized by the iterative function $\textproc{SafeIt}$ given as pseudo-code in \REFalg{alg:SafeIt}. 

\begin{algorithm}[t]
	\caption{Procedure $\textproc{SafeIt}$}\label{alg:SafeIt}
	\begin{algorithmic}[1]
		\Require $\Aux\subseteq \Xa{1}$, $\Aux'\subseteq \Xa{1}$, $l$, $\Cset$
		\State \textcolor{gray}{$\textproc{Explore}(\UA{l1}(\Aux)\setminus\UA{l1}(\Aux'),l)$} \label{alg:SafeIt:Explore}
		\State $W \gets \FCpre{l}(\UA{l1}(\Aux)) \cap \UA{l1}(\Aux)$\label{alg:SafeIt:computeW}
		\State $\Cset\gets\Cset\cup\set{C_l \gets (W,\Ua{},\emptyset)}$\label{alg:SafeIt:safeB}
		\State $\Aux' \gets \Aux'\cup\UA{1l}(W)$\label{alg:SafeIt:safeW}
		\If{$l\neq 1$}
			\State $\textproc{SafeIt}(\Aux,\Aux',l-1,\Cset)$ // go finer
		\Else
			\If{$\Aux \neq \Aux'$}
				\State $\textproc{SafeIt}(\Aux',\emptyset,L,\emptyset)$ // start new iteration\label{alg:SafeIt:iterate}
			\Else
				\State \Return $\langle \Aux,\Cset \rangle$ // terminate \label{alg:SafeIt:termination}
			\EndIf
		\EndIf
	\end{algorithmic}
\end{algorithm}
In order to map abstract states between different layers of abstraction, $\textproc{SafeIt}$ uses the operator 
	\begin{equation}\label{equ:Gamma}
	 		\UA{ll'}(\Upsilon_{l'}) = 
			\begin{cases}
				\Ra{ll'}(\Upsilon_{l'}), & l \leq l' \\
				\SetComp{\hat{x} \in \Xa{l}}{\Ra{l'l}(\hat{x})\subseteq \Upsilon_{l'}}, & l > l' 
			\end{cases}
	\end{equation}
where $l,l'\in [1;L]$ and $\Upsilon_{l'}\subseteq \Xa{l'}$. 
The operation $\UA{ll'}(\cdot)$ under-approximates a set of layer $l'$ to a set of layer $l$. 

In contrast to \REFsec{sec:ML_Safety_Frontier}, we consider non-lazy synthesis in this section which assumes that $\Sa{l}$ is pre-computed for all states within the safe set in every \layer before $\textproc{SafeIt}$ is called. This can be formalized by a wrapper function $\textproc{EagerSafe}(\Tseta{1},L)$ which first calls $\textproc{Explore}(\UA{l1}(\Tseta{1}),l)=\textproc{Explore}(\Tseta{l},l)$ (see \REFalg{alg:Explore}) for every \layer and then calls $\textproc{SafeIt}(\Tseta{1},\emptyset,L,\emptyset)$.

%


\begin{algorithm}[t]
	\caption{$\textproc{Explore}$}\label{alg:Explore}
	\begin{algorithmic}[1]
		\Require $\Aux\subseteq \Xa{l}$, $l$
		\For {$\xa\in\Aux,\ua\in\Ua$}
		\If {$\Fa{l}(\xa,\ua)$ is undefined}
		\State compute $\Fa{l}(\xa,\ua)$ as in \eqref{eq:next state abs sys 0}
		\EndIf
		\EndFor
	\end{algorithmic}
\end{algorithm}
Due to the monotonic nature of the iterative computation of safe sets, the set $\Aux$ in \REFalg{alg:SafeIt} is always a subset of $\Tseta{1}$ (see \REFlem{lem:monotonicity} for a formal proof). This implies that line~\ref{alg:SafeIt:Explore} of \REFalg{alg:SafeIt} (indicated in gray) will never perform any exploration (as all needed transition relations are pre-computed) and can therefore be ignored in this section.

When initialized with $\textproc{SafeIt}(\Tseta{1},\emptyset,L,\emptyset)$, \REFalg{alg:SafeIt} performs the following computations: it starts in layer $l=L$ with an outer recursion count $i=1$ (not shown in \REFalg{alg:SafeIt}) and reduces $l$, one step at the time, until $l=1$ is reached, at which point it then starts over again from layer $L$ with $i=i+1$ and a new safe set $\Aux'$. In every such iteration $i$, one step of the safety fixed-point is performed for every layer and the resulting set is stored in the layer $1$ map $\Aux'\subseteq \Xa{1}$, whereas $\Aux\subseteq \Xa{1}$ keeps the knowledge of the previous iteration. If the finest layer is reached and we have $\Aux=\Aux'$, the algorithm terminates. Otherwise $\Aux'$ is copied to $\Aux$, $\Aux'$ and $\Cset$ are reset to $\emptyset$ and $\textproc{SafeIt}$ starts a new iteration (see line~\ref{alg:SafeIt:iterate}).
After $\textproc{SafeIt}$ has terminated, it returns a multi-layered controller $\Cset=\set{\Ci{l}}_\layer$ which only contains the domains of the respective controllers $\Ci{l}$ for every layer (see line~\ref{alg:SafeIt:termination}). The transition functions $\Gci{l}$ can be computed by choosing one input $\ua\in\Ua$ for every $\xa\in \Uci{l}$ s.t.\ 
\begin{equation}\label{def:safecontrollerGci}
 \propImp{\ua=\Gci{l}(\xa)}{\Fa{l}(\xa,\ua)\subseteq \UA{l1}(\Aux)}.
\end{equation}


Note that states encountered for layer $l$ in iteration $i$ are saved to the lowest layer $1$ (line~\ref{alg:SafeIt:safeW} of \REFalg{alg:SafeIt}) and \enquote{loaded} back to the respective layer $l$ in iteration $i+1$ (line~\ref{alg:SafeIt:computeW} of \REFalg{alg:SafeIt}). Therefore, a state $\xa\in\Xa{l}$ with $l>1$, which was not contained in $W$ as computed in layer $l$ and iteration $i$ via line~\ref{alg:SafeIt:computeW} of \REFalg{alg:SafeIt}, might still be included in $\UA{l1}(\Aux)$ loaded in the next iteration $i+1$ when re-computing line~\ref{alg:SafeIt:computeW} for $l$. This happens if all states $x\in\xa$ were added to $\Aux'$ by some layer $l'<l$ in iteration $i$. 
This allows the algorithm to \enquote{bridge} regions that require a finer grid and to use layer $L$ in all remaining regions of the state space. 
The latter is not true for the multi-layered safety algorithm given in \cite{HsuMajumdarMallikSchmuck_HSCC18}, Sec.~4.3 as shown by the example in \REFsec{sec:spiral}.


\smallskip
\noindent\textbf{Soundness and Relative Completeness}\footnote{Absolute completeness of controller synthesis cannot be guaranteed by ABCS; we therefore provide completeness relative to the finest layer.}
Due to the effect described above, the map $W$ encountered in line~\ref{alg:SafeIt:computeW} for a particular layer $l$ throughout different iterations $i$ might not be monotonically shrinking. However, the latter is true for layer $1$, which is formalized by the following lemma.

\begin{lemma}\label{lem:monotonicity}
 Let $\Aux^0:=\Tseta{1}$ and let $\textproc{SafeIt}$ be called by $\textproc{EagerSafe}(\Tseta{1},L)$, terminating after $N$ iterations. Further, set $\Aux^i:=\Aux'$ whenever \REFalg{alg:SafeIt} reaches line~\ref{alg:SafeIt:safeW} with $l=1$ for the $i$-th time. Then it holds that $\Aux^i\subseteq\Aux^{i-1}$, hence $\Aux^i\subseteq\Tseta{1}$ for all $i\leq N$.
\end{lemma}

\begin{proof}
 Let $W^i_l$ be the set computed in line~\ref{alg:SafeIt:computeW} of \REFalg{alg:SafeIt} in the $i$-th iteration for $l$ and observe that 
 $W^i_l=\FCpre{l}(\UA{l1}(\Aux^{i-1})) \cap \UA{l1}(\Aux^{i-1})\subseteq\UA{l1}(\Aux^{i-1})$. Using \eqref{equ:Gamma} this implies $\UA{1l}(W^i_l)\subseteq\Aux^{i-1}$. As $\Aux^i=\bigcup_\layer\UA{1l}(W^i_l)$ we have $\Aux^i\subseteq\Aux^{i-1}$.
\end{proof}

This leads to our first main result, showing that $\textproc{EagerSafe}(\Tseta{1},L)$ is sound and relatively complete.

\begin{theorem}\label{thm:EagerSafe}
Let $\tuple{\Sigma,\Tset}$ be a safety control problem and $\Saset = \set{\Sa{l}}_\layer$ a sequence of abstractions. Let $\langle \Aux^N,\Cset \rangle=\textproc{EagerSafe}(\Tseta{1},L)$ s.t.\ $\Cset=\set{\Ci{l}}_\layer$ and $\Gci{l}$ is defined as in \eqref{def:safecontrollerGci} for all \layer. Further, let $\overline{B}$ be the domain of the single-layer safety controller computed using \eqref{equ:safe-fp} for $l=1$.
Then $\Cset\in\WIN(\Sigma,\Tset{})$ and $\overline{B}\subseteq\Aux^N$, i.e. $\Cset$ is sound and relatively complete w.r.t.\ single-layer control for layer $l=1$. 
\end{theorem}

\begin{proof}
To prove \emph{soundness}, i.e., $\Cset\in\WIN(\Sigma,\Tset{})$, we show that \eqref{equ:Csetsound} holds. This is true if for all \layer and $\xa\in\Uci{l}$, it holds that (i) $\xa\in\Tseta{l}$ and (ii) there exists $\ua\in\Gci{l}(\xa)$ s.t. $\Fa{l}(\xa,\ua) \neq \emptyset$, and for all $\xa''\in\Fa{l}(\xa,\ua)$,\ $\Qset(\Qai{ l}(\xa''))\neq\emptyset$. 
As \REFlem{lem:monotonicity} implies $W^i_l\subseteq\UA{l1}(\Aux^{i-1})$, $\Aux^i\subseteq\Tseta{1}$ and $\UA{l1}(\Tseta{1})=\Tseta{l}$ we have $W^i_l\subseteq\Tseta{l}$. Further, line~\ref{alg:SafeIt:safeB} of \REFalg{alg:SafeIt} implies $\Uci{l}=W^N_l\subseteq\Tseta{l}$, proving (i). As $\Aux^N=\Aux^{N-1}$  line~\ref{alg:SafeIt:computeW} of \REFalg{alg:SafeIt} implies $\Uci{l}=W^N_l\subseteq\FCpre{l}(\UA{l1}(\Aux^N))$. Hence there is $\ua$ s.t.\ \eqref{def:safecontrollerGci} holds (from \eqref{eq:define cpre}), implying that $\Fa{l}(\xa,\ua)\neq \emptyset$. It follows from the definition of $\Qset$ that $\Qset(\Qai{ l}(\xa''))\neq\emptyset$.

We prove \emph{completeness}, i.e., $\overline{B}\subseteq\Aux^N$, by induction in $i\in [1;N]$. Given $\widetilde{W}^i$ as in \eqref{equ:safe-fp} for $l=1$ we show $\widetilde{W}^i\subseteq\Aux^i$. The base case is $\widetilde{W}^0=\Tseta{1}=\Aux^0$. The induction step gives
\begin{align*}
 \widetilde{W}^i&=\FCpre{1}(\widetilde{W}^{i-1}) \cap \widetilde{W}^{i-1}\subseteq\FCpre{1}(\Aux^{i-1}) \cap \Aux^{i-1}\\
 &\subseteq \textstyle\cup_\layer(\FCpre{l}(\UA{l1}(\Aux^{i-1})) \cap \UA{l1}(\Aux^{i-1}))=\Aux^i.
\end{align*}
With this, it obviously holds that $\overline{B}=\widetilde{W}^{\widetilde{N}}\subseteq\Aux^N$.
\end{proof}

\section{Computing Abstractions Lazily}\label{sec:ML_Safety_Frontier}

As our main contribution, we now consider the case where the multi-layered abstraction $\Saset$ is not pre-computed. This is implemented by $\textproc{LazySafe}(\Tseta{1},L)$ which simply calls $\textproc{SafeIt}_1(\Tseta{1},\emptyset,L,\emptyset)$. 
%
With this, line~\ref{alg:SafeIt:Explore} of \REFalg{alg:SafeIt} is used to explore transitions in all states in layer $l$ which are (i) not marked unsafe by all layers in the previous iteration, i.e., are in $\UA{l1}(\Aux)$, but (ii) cannot stay safe for $i$ times-steps in any layer $l'>l$, i.e., are not in $\UA{l1}(\Aux')$. 
In the first iteration of $\textproc{SafeIt}(\Tseta{1},\emptyset,L,\emptyset)$ the set $\UA{l1}(\Aux)\setminus\UA{l1}(\Aux')$ is same as $\UA{1L}(\Tseta{1})=\Tseta{L}$. Hence, for layer $L$ all transitions for states inside the safe set are pre-computed in the first iteration of \REFalg{alg:SafeIt}. 
This is in general not true for lower layers $l<L$. 

To ensure that the lazy exploration of the state space is still sound and relatively complete, we show in the following lemma that all states which need to be checked for safety in layer $l$ of iteration $i$ are indeed explored.


\begin{lemma}\label{lem:EffectLazyness}
Let $W^i_l$ and $\Aux'^i_l$ (resp. $\widetilde{W}^i_l$ and $\widetilde{\Aux}'^i_l$) denote the set computed in line~\ref{alg:SafeIt:computeW} and \ref{alg:SafeIt:safeW} of \REFalg{alg:SafeIt} (called by $\textproc{EagerSafe}(\Tseta{1},L)$ and $\textproc{LazySafe}(\Tseta{1},L)$, resp.) the $i$-th time for $l$.  
Then it holds for all $i\in [1;N]$ and \layer that
 \begin{align}\label{equ:tildeWpl}
  \widetilde{W}^i_l\subseteq W^i_l\quad\text{and}\quad \Aux'^i_{l}=\widetilde{\Aux}'^i_{l}.
 \end{align}
\end{lemma}

\begin{proof}
First observe that the algorithm \REFalg{alg:SafeIt} (called by $\textproc{EagerSafe}(\Tseta{1},L)$ and $\textproc{LazySafe}(\Tseta{1},L)$, resp.) starts with $i=1$ and $l=L$. It first decrements $l$ (while keeping $i$ constant) until $l=1$ is reached, and then increments $i$ to $i+1$ and resets $l=1$ to $l=L$. We prove invariance of $W$ and $\Aux'$ to both steps separately, to show that \eqref{equ:tildeWpl} holds.

First consider the incrementation of $i$ in line~\ref{alg:SafeIt:iterate} of \REFalg{alg:SafeIt}. This implies that $\Aux'$ and $\widetilde{\Aux}'$ are copied to $\Aux$ and $\widetilde{\Aux}$. Hence, given the notation of this lemma, we have $\Aux^{i+1}_L=\Aux'^i_1$ and $\widetilde{\Aux}^{i+1}_L=\widetilde{\Aux}'^i_1$. Given this observation we do an induction over $i$ to show that $\widetilde{W}^{i+1}_L= W^{i+1}_L$ and $\Aux'^{i+1}_{L}=\widetilde{\Aux}'^{i+1}_{L}$ are true for all $i$. For $i=1$ (base case), $\Aux'^1_{1}=\widetilde{\Aux}'^1_{1}=\Tseta{1}$, which in turn implies $\widetilde{W}^{1}_L= W^{1}_L$.
%
For the induction step, observe that 
the induction hypothesis implies $\Aux'=\widetilde{\Aux}'$ and $\Aux=\widetilde{\Aux}$ in the right side of line~\ref{alg:SafeIt:Explore}-\ref{alg:SafeIt:safeW} of \REFalg{alg:SafeIt}, no matter if \REFalg{alg:SafeIt} is called by $\textproc{EagerSafe}(\Tseta{1},L)$ or $\textproc{LazySafe}(\Tseta{1},L)$. This implies $\widetilde{W}^{i+1}_L= W^{i+1}_L$ (computed in line \ref{alg:SafeIt:computeW}) and $\Aux'^{i+1}_{L}=\widetilde{\Aux}'^{i+1}_{L}$ (updated in line~\ref{alg:SafeIt:safeW}) whenever the claim holds for $i$. 

Second, we consider decrementing $l$ while keeping $i$ constant. We do an induction over $l$ by assuming $\Aux'^{i}_{l+1}=\widetilde{\Aux}'^{i}_{l+1}$ and show that this implies $\widetilde{W}^i_l\subseteq W^i_l$ and $\Aux'^i_{l}=\widetilde{\Aux}'^i_{l}$. Observe that the base case is $l+1=L$, which was established by the induction over $i$. 
For the induction step over $l$, observe that the induction assumption implies $\UA{l1}(\widetilde{\Aux}^{i-1})\setminus\UA{l1}(\widetilde{\Aux}'^i_{l+1})=\UA{l1}(\Aux^{i-1})\setminus\UA{l1}(\Aux'^i_{l+1})$ and, as  $\Aux^i\subseteq\Tseta{1}$ (from \REFlem{lem:monotonicity}), we have $\UA{l1}(\Aux^{i-1})\setminus\UA{l1}(\Aux'^i_{l+1})\subseteq\UA{l1}(\Tseta{1})$. Further, observe that $\textproc{EagerSafe}(\Tseta{1},L)$ explores $\UA{l1}(\Tseta{1})$ once, while $\textproc{LazySafe}(\Tseta{1},L)$ explores $\UA{l1}(\Aux^{i-1})\setminus\UA{l1}(\Aux'^i_{l+1})$ via line~\ref{alg:SafeIt:Explore} of \REFalg{alg:SafeIt} in every iteration $i$. This implies that for the computation of $\widetilde{W}^i_l$ in line~\ref{alg:SafeIt:iterate} of \REFalg{alg:SafeIt} via \eqref{eq:define cpre} the transition function $\Fa{l}(\xa,\ua)$ is computed for a subset of states compared to the computation of $W^i_l$. With this it immediately follows from \eqref{eq:define cpre} that $\widetilde{W}^i_l\subseteq W^i_l$.
 To show the right side of \eqref{equ:tildeWpl}, recall that $\Fa{l}(\xa,\ua)$ is at least computed for the set $\UA{l1}(\Aux^{i-1})\setminus\UA{l1}(\Aux'^i_{l+1})$ via line~\ref{alg:SafeIt:Explore} of \REFalg{alg:SafeIt} (and possibly for some more states which were explored in previous iterations) when $\widetilde{W}^i_l$ is computed. Using \eqref{eq:define cpre} this implies that 
$
 \widetilde{W}^i_l\supseteq\left(\FCpre{l}(\UA{l1}(\Aux^{i-1}))\setminus\UA{l1}(\Aux'^i_{l+1})\right)\cap\UA{l1}(\Aux^{i-1})
 = W^i_l\setminus\UA{1l}(\Aux'^i_{l+1}).$
With this we have 
$\Aux'^i_{l}
=\Aux'^i_{l+1}\cup \UA{1l}(W^i_l)
=\Aux'^i_{l+1}\cup \left(\UA{1l}(W^i_l)\setminus\Aux'^i_{l+1}\right)
=\Aux'^i_{l+1}\cup\UA{1l}\left(W^i_l\setminus\UA{l1}(\Aux'^i_{l+1})\right)
\subseteq\Aux'^i_{l+1}\cup\UA{1l}(\widetilde{W}^i_l)
=\widetilde{\Aux}'^i_{l+1}\cup\UA{1l}(\widetilde{W}^i_l)
=\widetilde{\Aux}'^i_l$ 
and 
$\widetilde{\Aux}'^i_l
=\widetilde{\Aux}'^i_{l+1}\cup\UA{1l}(\widetilde{W}^i_l)
=\Aux'^i_{l+1}\cup\UA{1l}(\widetilde{W}^i_l)
\subseteq\Aux'^i_{l+1}\cup \UA{1l}(W^i_l)
=\Aux'^i_{l}$, which completes the induction step over $l$.
\end{proof}

Now as direct consequence of \REFthm{thm:EagerSafe} and \REFlem{lem:EffectLazyness}, we present our second main result:


\begin{theorem}\label{thm:LazySafe}
Let $\tuple{\Sigma,\Tset}$ be a safety control problem and $\Saset = \set{\Sa{l}}_\layer$ a sequence of abstractions. Let $\langle \Aux^N,\Cset \rangle=\textproc{LazySafe}(\Tseta{1},L)$ s.t.\ $\Cset=\set{\Ci{l}}_\layer$ and $\Gci{l}$ be defined as in \eqref{def:safecontrollerGci} for all \layer.
Further, let $\overline{B}$ be the domain of the single-layer safety controller computed using \eqref{equ:safe-fp} for $l=1$.
Then $\Cset\in\WIN(\Sigma,\Tset{})$ and $\overline{B}\subseteq\Aux^N$, i.e., $\Cset$ is sound and relatively complete w.r.t.\ single-layer control for layer $l=1$. 
\end{theorem}

\begin{proof}
 First recall that \REFlem{lem:EffectLazyness} implies $\Aux^i=\widetilde{\Aux}^i$ for all $i\leq N$. Therefore \REFlem{lem:monotonicity} equivalently holds for $\widetilde{\Aux}^i$ and the completeness proof of \REFthm{thm:EagerSafe} is equivalent to the one of \REFthm{thm:LazySafe}. For the soundness proof, observe that \eqref{equ:tildeWpl} implies $B^l=\widetilde{W}^i_l\subseteq W^i_l\subseteq\Tseta{l}$ and $B^l=\widetilde{W}^i_l\subseteq W^i_l\subseteq\FCpre{l}(\UA{l1}(\Aux^N))$, from which (i) and (ii) follows. 
\end{proof}

\section{Experimental Evaluation} \label{sec:experiments}

We evaluate our algorithm on a benchmark DC-DC boost converter example from \cite{GirardPolaTabuada_2010, mouelhi2013cosyma, Scots}. The system $\Sigma$ is a second order differential inclusion $\dot{X}(t) \in A_pX(t) + b + W$ with two switching modes $p\in \set{1,2}$, where
\begin{align*}
	b = \begin{bmatrix}
		\frac{v_s}{x_l}\\
		0
	\end{bmatrix}, 
	A_1 = \begin{bmatrix}
		-\frac{r_l}{x_l}  &  0\\
		0 				  &  -\frac{1}{x_c}\frac{r_0}{r_0+r_c}
	\end{bmatrix},\\
	A_2 = \begin{bmatrix}
		-\frac{1}{x_l}(r_l+\frac{r_0r_c}{r_0+r_c})	&	\frac{1}{5}(-\frac{1}{x_l}\frac{r_0}{r_0+r_c})\\
		5\frac{r_0}{r_0+r_c}\frac{1}{x_c}			&	-\frac{1}{x_c}\frac{1}{r_0+r_c}
	\end{bmatrix},
\end{align*}
with $r_0 = 1$, $v_s = 1$, $r_l = 0.05$, $r_c = 0.5r_l$, $x_l = 3$, $x_c = 70$ and $W = [-0.001, 0.001]\times [-0.001, 0.001]$. A physical and more detailed description of the model can be found in \cite{GirardPolaTabuada_2010}. The safety control problem that we consider is given by $\langle \Sigma,T \rangle$, where $T = [1.15,1.55]\times [5.45,5.85]$. We evaluate the performance of our $\textproc{LazySafe}$ algorithm on this benchmark and compare it \begin{inparaenum} \item to the one presented in \cite{HsuMajumdarMallikSchmuck_HSCC18}, Sec.~4.3 which we call \textproc{ML\_Safe}, and \item to the single-layered version of \texttt{SCOTS}. \end{inparaenum} For $\textproc{LazySafe}$ and \textproc{ML\_Safe}, we vary the number of layers used. The results are presented in \REFfig{fig:runtimes}. The finest layer is common to each trial and is parameterized by $\eta_1 = [0.0005, 0.0005], \ \tau_1 = 0.0625$, with the ratio between the grid parameters and sampling times of successive layers being $2$. All experiments presented in this section were performed with a system equipped with an Intel\textsuperscript{\textregistered} Core\textsuperscript{\texttrademark} i5-6600 3.30GHz CPU and 16GB of RAM.

From \REFfig{fig:runtimes}, we see that $\textproc{LazySafe}$ is significantly faster than both \textproc{ML\_Safe} and \texttt{SCOTS} for higher value of $L$. 
The single layered case ($L=1$) takes slightly more time in both $\textproc{LazySafe}$ and \textproc{ML\_Safe} than \texttt{SCOTS} due to some extra overhead of the multi-layered algorithm. 

In \REFfig{fig:dcdc}, we visualize the domain of the constructed transitions and the synthesized controllers in each layer for $\textproc{LazySafe}(\cdot,6)$. The safe set is mostly covered by cells in the two coarsest layers. This phenomenon is responsible for the computational savings over $\textproc{LazySafe}(\cdot,1)$. 

\begin{figure}[h]
	\centering
	\begin{tikzpicture}[scale=.9,point/.style={circle,scale=0.2,draw=black,fill=red,thick}]	
	\begin{axis}[
	    	legend cell align={left},
		ybar stacked,
		bar shift=3pt,
		xlabel= No. of layers ($L$),
		ylabel= Runtime in seconds,
		axis y line=left,
		axis x line=bottom,
		xmin=0, xmax=8,
		ymin=0, ymax=500,
		xtick={0,1,2,3,4,5,6,7},
		ytick={0,100,200,300,400,500},
		xticklabel=$\pgfmathprintnumber{\tick}$,
		yticklabel=$\pgfmathprintnumber{\tick}$,
		legend style={at={(0.5,1)},anchor=north west,draw=none},
        ]
        \addplot[fill=purple]	
			coordinates {
        		(1, 28.355)
        		(2, 9.59962)
        		(3, 5.98247)
        		(4, 5.54938)
        		(5, 4.39672)
        		(6, 4.48306)
        		(7, 4.45758)	
        	};
			\addlegendentry{\textproc{LazySafe} abstraction}
		
        \addplot[fill=blue]	
			coordinates {
        		(1, 410.783)
        		(2, 75.2443)
        		(3, 31.1342)
        		(4, 14.5017)
        		(5, 13.8379)
        		(6, 14.0392)
        		(7, 14.2794)	
        	};
        	\addlegendentry{\textproc{LazySafe} synthesis};
        
	\end{axis}
	
	\begin{axis}[
	    	legend cell align={left},
		ybar stacked,
		bar shift=-3pt,
		xlabel= No. of layers ($L$),
		ylabel= Runtime in seconds,
		axis y line=left,
		axis x line=bottom,
		xmin=0, xmax=8,
		ymin=0, ymax=500,
		xtick={0,1,2,3,4,5,6,7},
		ytick={0,100,200,300,400,500},
		xticklabel=$\pgfmathprintnumber{\tick}$,
		yticklabel=$\pgfmathprintnumber{\tick}$,
		legend style={at={(0.5,0.8)},anchor=north west,draw=none},
        ]
        
        \addplot[fill=red!30]	
			coordinates {
        		(1, 28.8522)
        		(2, 37.5347)
        		(3, 41.3946)
        		(4, 41.415)
        		(5, 0)
        		(6, 0)
        		(7, 0)	
        	};
            \addlegendentry{\textproc{ML\_Safe} abstraction}
            
        \addplot[fill=black!30]	
			coordinates {
        		(1, 417.597)
        		(2, 50.0761)
        		(3, 38.5002)
        		(4, 36.1657)
        		(5, 0)
        		(6, 0)
        		(7, 0)	
        	};
            \addlegendentry{\textproc{ML\_Safe} synthesis};
	\end{axis}
	
	\begin{axis}[
	    	legend cell align={left},
		ybar stacked,
		bar shift=-9pt,
		xlabel= No. of layers ($L$),
		ylabel= Runtime in seconds,
		axis y line=left,
		axis x line=bottom,
		xmin=0, xmax=8,
		ymin=0, ymax=500,
		xtick={0,1,2,3,4,5,6,7},
		ytick={0,100,200,300,400,500},
		xticklabel=$\pgfmathprintnumber{\tick}$,
		yticklabel=$\pgfmathprintnumber{\tick}$,
		legend style={at={(0.5,0.6)},anchor=north west,draw=none},
        ]
        \addplot[fill=yellow!50]	
			coordinates {
        		(1, 26.7736)
        		(2, 0)
        		(3, 0)
        		(4, 0)
        		(5, 0)
        		(6, 0)
        		(7, 0)	
        	};
			\addlegendentry{\textproc{SCOTS} abstraction}
		
        \addplot[fill=green!50]	
			coordinates {
        		(1, 357.999)
        		(2, 0)
        		(3, 0)
        		(4, 0)
        		(5, 0)
        		(6, 0)
        		(7, 0)	
        	};
        	\addlegendentry{\textproc{SCOTS} synthesis};
        
	\end{axis}
\end{tikzpicture}	
	\vspace{-0.2cm}
	\caption{Comparison of run-time of \textproc{LazySafe}, \textproc{ML\_Safe} and \texttt{SCOTS} on the DC-DC boost converter example. $L\geq 5$ was not used for \textproc{ML\_Safe} since more coarser layers failed to produce a non-empty winning set. The same was true for $L\geq 8$ for \textproc{LazySafe}.}	
	\label{fig:runtimes}
\end{figure}
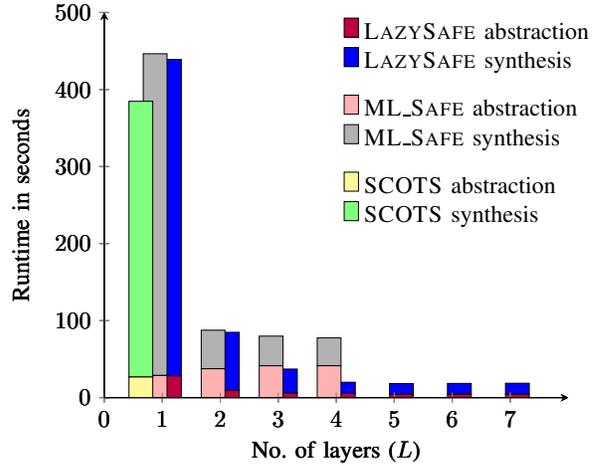

\begin{figure}[h]
	\centering
	\includegraphics[width=1\columnwidth]{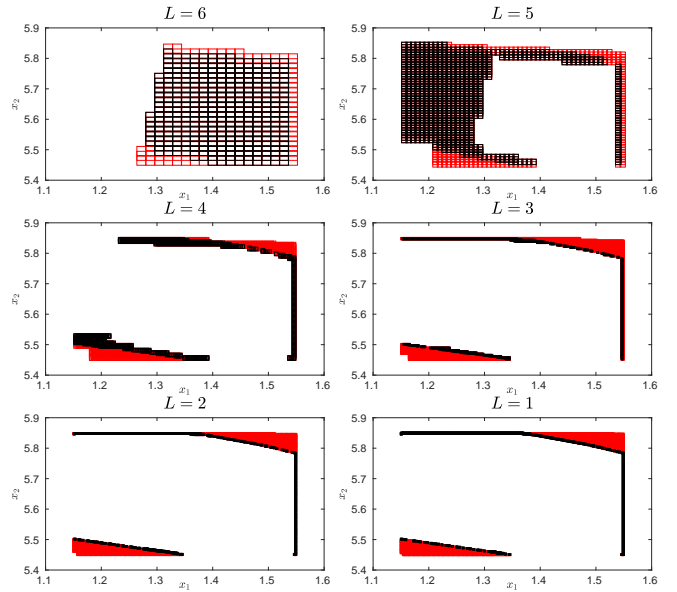}
	\vspace{-0.3cm}
	\caption{Domain of the computed transitions (union of red and black region) and the synthesized controllers (black region) for the DC-DC boost converter example, computed by $\textproc{LazySafe}(\cdot,6)$.}
	\label{fig:dcdc}
\end{figure}

 

\bibliographystyle{abbrv}

\end{document}